\documentclass{article}

\usepackage{amsmath, amsthm, amssymb}
\usepackage[dvipdfm]{graphicx}
\usepackage{comment}

\theoremstyle{plain}
\newtheorem{lemma}{Lemma}
\newtheorem{proposition}{Proposition}
\newtheorem{theorem}{Theorem}
\newtheorem{corollary}{Corollary}

\theoremstyle{remark}
\newtheorem{example}{Example}

\newcommand{\LHC}{\ensuremath{\mathcal{L\hspace*{-0.5mm}H}}}
\newcommand{\RHC}{\ensuremath{\mathcal{R\hspace*{-0.5mm}H}}}
\newcommand{\HC}{\ensuremath{\mathcal{H}}}

\newcommand{\Pref}{\ensuremath{\mathrm{Pref}}}
\newcommand{\Suff}{\ensuremath{\mathrm{Suff}}}

\newcommand{\ind}{\ensuremath{\mathrm{ind}}}
\newcommand{\maxind}{\ensuremath{\mathrm{maxind}}}

\newcommand{\calpha}{\ensuremath{\overline{\alpha}}}

\title{On the regularity of iterated hairpin completion of a single word}
\author{Lila Kari, Steffen Kopecki, and Shinnosuke Seki}

\begin{document}

\maketitle

\begin{abstract}
	Hairpin completion is an abstract operation modeling a DNA bio-operation which receives as input a DNA strand $w = x\alpha y \calpha$, and outputs $w' = x \alpha y \bar{\alpha} \overline{x}$, where $\overline{x}$ denotes the Watson-Crick complement of $x$. 
	In this paper, we focus on the problem of finding conditions under which the iterated hairpin completion of a given word is regular. 
	According to the numbers of words $\alpha$ and $\calpha$ that initiate hairpin completion and how they are scattered, we classify the set of all words $w$. 
	For some basic classes of words $w$ containing small numbers of occurrences of $\alpha$ and $\calpha$, we prove that the iterated hairpin completion of $w$ is regular. 
	For other classes with higher numbers of occurrences of $\alpha$ and $\calpha$, we prove a necessary and sufficient condition for the iterated hairpin completion of a word in these classes to be regular. 
\end{abstract}

	\section{Introduction}

\begin{figure}
\begin{center}
\includegraphics{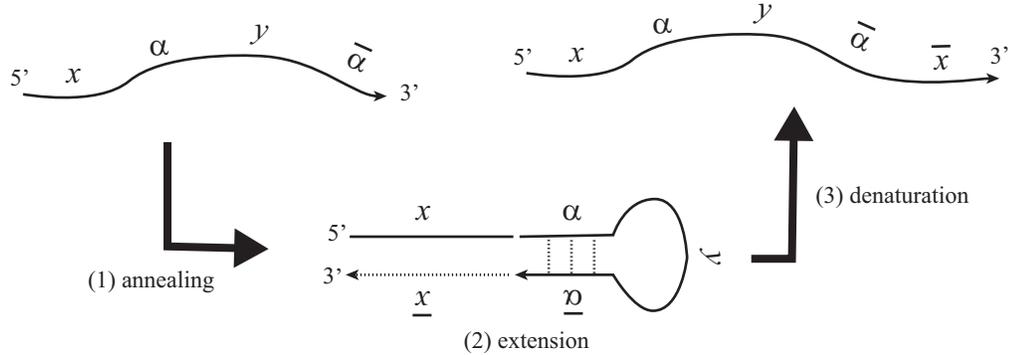}
\caption{
	Hairpin completion by polymerase chain reaction \cite{HAKSY00, SKKGYISH99}. 
	The operation input is $x\alpha y \calpha$, the output is $x\alpha y \bar{\alpha} \overline{x}$, and the primer is $\alpha$. 
}
\label{fig:selfPCR}
\end{center}
\end{figure}

A DNA strand can be abstractly viewed as a word over the alphabet $\{{\tt A}, {\tt C}, {\tt G}, {\tt T}\}$, where in {\tt A} is Watson-Crick complementary to {\tt T} and {\tt C} to {\tt G}, and two complementary DNA single strands of opposite orientation bind together to form a double DNA strand (intermolecular structure). 
Also, if subwords of a DNA strand are complementary, the strand may bind to itself forming intramolecular structures such as stem-loops, also known more commonly as {\it hairpins} (Figure~\ref{fig:selfPCR} (2)). 
Hairpins can be a building block of a larger-scale structure of RNA strands, and play a role in determining various chemical and thermodynamical properties (stability, structures, functions) of the structure, and make significant contributions to the genetic information processing as illustrated in their function as a stopper for messenger RNA (mRNA) transcription. 
A {\tt CG}-rich sequence of an mRNA folds into its Watson-Crick complement on the RNA and forms a stable hairpin. 
Transcription of the mRNA is terminated when RNA polymerase reaches the hairpin. 
At that time, {\it nusA} protein bound to the polymerase interacts with the hairpin and takes the polymerase off the mRNA. 
This hairpin-driven mechanism is called {\it intrinsic termination} \cite{WilsonHippel95}. 
As such, hairpins tend to interfere with reactions, and therefore were given the cold shoulder by DNA computing experimentalists. 
See \cite{Adleman94, AritaKobayashi02, JoKeMa02, JonoskaMahalingam04, KaKoLoSoTh06, PaRoYo01} about this problem and about some of the ``good'' designs of DNA strands that are free of hairpins. 

Hairpin is not a foe to all DNA computing experiments; many molecular computing machineries have been proposed which make good use of hairpins. 
Such hairpin-driven systems include DNA RAM \cite{KaYaOhYaHa08, TakinoueSuyama04, TakinoueSuyama06} and Whiplash PCR \cite{HAKSY00, SKKGYISH99}. 
In particular, Whiplash PCR features a self-directed polymerase chain reaction (PCR) of DNA strand, which practically motivates the investigation of a formal language operation called {\it hairpin completion}. 
Hairpin completion proceeds as follows (Figure~\ref{fig:selfPCR}): 
Starting from a DNA strand $w = x\alpha y \calpha$, a segment $\calpha$ at the 3'-end of $w$ binds to its Watson-Crick complementary strand $\alpha$ on the strand (annealing). 
A polymerase chain reaction then extends $w$ at its 3'-end in the $5' \to 3'$ direction so as to generate the strand $x\alpha y \bar{\alpha} \overline{x}$ (let us call $\alpha$ and $\calpha$ that bind with each other to initiate this PCR reaction {\it primers}). 
Despite the intrinsic $5' \to 3'$ polarity of polymerases, a mechanism exists to make polymerase reaction work in the $3' \to 5'$ direction (Okazaki fragment \cite{OkOkSaSuSu68}). 

As an abstract model of the above-mentioned self-directed PCR, Cheptea, Mart\'{i}n-Vide, and Mitrana proposed the hairpin completion in \cite{ChMaMi06}, and since then this abstract operation has been studied on its algorithmic and formal linguistic aspects \cite{DiekertKopecki11, MaViMi09, MaMiYo09, MaMiYo10} together with its variant called {\it bounded hairpin completion} \cite{ItLeMaMi11, Kopecki10}, where the length of extension in one operation is bounded by a constant. 
Ito et al.~\cite{ItLeMaMi11} and Kopecki~\cite{Kopecki10} proved that all classes in the Chomsky Hierarchy are closed under iterated bounded hairpin completion. 
In contrast, the class of regular languages was proved not to be closed under iterated (unbounded) hairpin completion \cite{ChMaMi06}, and a surprising fact is that iterated hairpin completion of {\it a} word can be non-regular \cite{Kopecki10}. 
In this paper, we focus on a problem proposed by Kopecki in \cite{Kopecki10}; is it decidable whether the iterated hairpin completion of a given word is regular? 
The iterated hairpin completion of a singleton language (a word) is known to be in {\tt NL} \cite{ChMaMi06}, but can be non-regular as shown in the following example. 

\begin{example}\label{ex:nonx_cs}
	Let $\alpha = a^k$ and $w = \alpha b \alpha c \alpha \bar{\alpha} \bar{d} \bar{\alpha}$, where $a, \bar{a}, b, \bar{b}, c, \bar{c}, d, \bar{d}$ are all distinct letters. 
	Then the intersection of the iterated hairpin completion of $w$ with $(\alpha b \alpha c (\alpha b)^+ \alpha d)^2 \alpha b \alpha c \alpha \bar{\alpha} \bar{d} \bar{\alpha} (\bar{b} \bar{\alpha})^+ \bar{c} \bar{\alpha} \bar{b} \bar{\alpha}$ is $\{(\alpha b \alpha c (\alpha b)^i \alpha d)^2 \alpha b \alpha c \alpha \bar{\alpha} \bar{d} \bar{\alpha} (\bar{b} \bar{\alpha})^i \bar{c} \bar{\alpha} \bar{b} \bar{\alpha} \mid i \ge 1 \}$. 
	This intersection is not context-free, and neither is the iterated hairpin completion. 
\end{example}

In this paper, we give a partial answer to the regularity-test decidability problem. 
We focus our attention on the number of primers a given word contains as its factors and on how these primers are scattered over the given word. 
All the words are classified in accordance with these two criteria, and for some basic classes, we give a necessary and sufficient condition for the iterated hairpin completion of a word in the class to be regular.

	\section{Preliminaries}

Let $\Sigma$ be an alphabet, $\Sigma^*$ be the set of all words over $\Sigma$, and for an integer $k \ge 0$, $\Sigma^k$ be the set of all words of length $k$ over $\Sigma$. 
The word of length 0 is called the empty word, denoted by $\lambda$, and let $\Sigma^+ = \Sigma^* \setminus \{\lambda\}$. 
A subset of $\Sigma^*$ is called a language over $\Sigma$. 
For a word $w \in \Sigma^*$, we employ the notation $w$ when we mean the word as well as the singleton language $\{w\}$ unless confusion arises. 
For a language $L \subseteq \Sigma^*$, we denote by $L^*$ the set $\{w_1 \cdots w_n \mid n \ge 0, w_1, \ldots, w_n \in L\}$. 

We equip $\Sigma$ with a function $\bar{\hspace{2mm}}: \Sigma \to \Sigma$ satisfying $\forall a \in \Sigma, \overline{\overline{a}} = a$; such a function is called an {\it involution}. 
This involution $\bar{\hspace*{2mm}}$ is naturally extended to words as: for $a_1, a_2, \ldots, a_n \in \Sigma$, $\overline{a_1 a_2 \cdots a_n} = \overline{a_n} \cdots \overline{a_2} \ \overline{a_1}$. 
For example, over the 4-letter alphabet $\Delta = \{{\tt A, C, G, T}\}$, if we define an involution $d: \Delta \to \Delta$ as $d({\tt A}) = {\tt T}$ and $d({\tt C}) = {\tt G}$, then $d$, being thus extended, maps the Watson strand of a complete DNA double strand into its Crick strand. 
The involution $d$ is called the Watson-Crick involution \cite{KariMahalingam08}. 
For a word $w \in \Sigma^*$, we call $\overline{w}$ the {\it complement of $w$}, being inspired by this application. 
A word $w \in \Sigma^*$ is called a {\it pseudo-palindrome} if $w = \overline{w}$. 
For a language $L \subseteq \Sigma^*$, $\overline{L} = \{\overline{w} \mid w \in L\}$. 

For words $u, w \in \Sigma^*$, if $w = xuy$ holds for some words $x, y \in \Sigma^*$, then $u$ is called a {\it factor} of $w$; a factor that is distinct from $w$ is said to be {\it proper}. 
If the equation holds with $x = \lambda$ ($y = \lambda$), then the factor $u$ is especially called a {\it prefix} (resp.~a {\it suffix}) of $w$. 
The prefix relation can be regarded as a partial order $\le_p$ over $\Sigma^*$; $u \le_p w$ means that $u$ is a prefix of $w$. 
Analogously, by $w \ge_s v$ we mean that $v$ is a suffix of $w$. 
For a word $w \in \Sigma^*$ and a language $L \subseteq \Sigma^*$, a factor $v$ of $w$ is {\it minimal with respect to $L$} if $v \in L$ and none of the proper factors of $v$ is in $L$. 

A nonempty word $w \in \Sigma^+$ is {\it primitive} if $w = x^i$ implies $i = 1$ for any nonempty word $x \in \Sigma^+$. 
It is well-known that for any nonempty word $w$, there exists a unique primitive word $u$ with $w \in u^+$. 
Such $u$ is called the {\it primitive root of $w$} and denoted by $\rho(w)$. 
Two words $x, y \in \Sigma^*$ {\it commute} if $xy = yx$, and this is known to be equivalent to $\rho(x) = \rho(y)$. 
See \cite{ChKa97} for details of primitivity and commutativity of words and related results. 

Now we introduce the operation investigated in this paper, that is, hairpin completion, and define it formally. 
Imagine that we have a DNA sequence $5'-{\tt CAATCGTATGAT}-3'$. 
The suffix ${\tt GAT}$ can find its $d$-image as a factor ${\tt ATC}$ on this sequence. 
Hence, this DNA sequence may bend over into a hairpin form by ${\tt GAT}$ binding with ${\tt ATC}$. 
This formation of hairpin structure leaves ${\tt CA}$ as a free sticky-end, and DNA polymerase converts it into the complete double strand by extending its 3'-end by ${\tt TG} = d({\tt CA})$. 
This exemplifies the mechanism of hairpin completion. 
We call two words whose thus binding initiate hairpin completion {\it primers}. 
In the above example, ${\tt GAT}$ and ${\tt ATC}$ are primers. 

Let $k$ be a constant that is assumed to be the length of a primer. 
Throughout this paper, we will not use the notation `$k$' for any other purpose. 
Let $\alpha \in \Sigma^k$ be a primer. 
If a given word $w \in \Sigma^*$ has a factorization $u \alpha v \overline{\alpha}$ for some $u, v \in \Sigma^*$ and $\alpha \in \Sigma^k$, then its {\it right hairpin completion with respect to $\alpha$} results in the word $u \alpha v \overline{\alpha} \bar{u}$. 
As long as $\alpha$ is clear from context, this operation is simply called {\it (single-primer) right hairpin completion}. 
By $w \to_{\RHC_\alpha} w'$, or by $w \to_{\RHC} w'$, we mean that $w'$ can be obtained from $w$ by right hairpin completion (with respect to $\alpha$). 
The {\it left hairpin completion} is defined analogously as an operation to derive $u' \alpha v' \overline{\alpha} \overline{u'}$ from $\alpha v' \overline{\alpha} \overline{u'}$, and the relation $\to_{\LHC_\alpha}$ is naturally introduced. 
By $\to_{\LHC}^*$ and $\to_{\RHC}^*$, we denote the reflexive transitive closure of $\to_{\LHC}$ and that of $\to_{\RHC}$, respectively. 
The relation $\to_{\HC}$ is defined as the union of $\to_{\LHC}$ and $\to_{\RHC}$. 

For a given language $L \subseteq \Sigma^*$, we define the set of words obtained by left hairpin completion from $L$, and the set of words obtained by iterated left hairpin completion from $L$, respectively, as follows: 
\[
	\LHC_\alpha(L) = \{w' \mid \exists w \in L, w \to_{\LHC_\alpha} w'\}, \hspace*{5mm} \LHC_\alpha^*(L) = \{w' \mid \exists w \in L, w \to_{\LHC_\alpha}^* w'\}. 
\]
Analogously, $\RHC_\alpha(L)$ and $\RHC_\alpha^*(L)$ are defined based on $\to_{\RHC}$ and $\to_{\RHC}^*$, and $\HC_\alpha(L)$ and $\HC_\alpha^*(L)$ are defined based on $\to_{\HC}$ and $\to_{\HC}^*$

\begin{proposition}\label{prop:presentation_RH}
	For a word $w \in \Sigma^*$, $\RHC_k^*(w) = \overline{\LHC_k^*(\overline{w})}$. 
\end{proposition}

	\section{Word structures relevant to the power of iterated hairpin completion}
	\label{sec:apref_casuf}

In this section, we describe several structural properties of a word $w$ that will be relevant for the characterization of its iterated hairpin completion $\HC_\alpha^*(w)$, where $\alpha \in \Sigma^k$ is a fixed parameter. 

A word $u \in \Sigma^*$ is called an {\it $\alpha$-prefix} of a word $w \in \Sigma^*$ if $w = u\alpha x$ for some word $x \in \Sigma^*$. 
In a similar manner, a word $v \in \Sigma^*$ is an {\it $\overline{\alpha}$-suffix} of $w$ if $w = y\overline{\alpha} v$ for some $y \in \Sigma^*$. 
If $w = y \overline{\alpha} v$ begins with $\alpha$, then this prefix can bind with the occurrence of $\overline{\alpha}$ (unless they overlap with each other), and left hairpin completion results in $\overline{v}w$. 
By $\Pref_\alpha(w)$ and $\Suff_{\overline{\alpha}}(w)$, we denote the set of all $\alpha$-prefixes and that of all $\overline{\alpha}$-suffixes of $w$, respectively. 
One can easily observe that $\Suff_{\calpha}(w) = \overline{\Pref_\alpha(\overline{w})}$. 
Throughout this paper, we let $\Pref_\alpha(w) = \{u_1, \ldots, u_m\}$ and $\Suff_{\calpha}(w) = \{\overline{v_1}, \ldots, \overline{v_n}\}$ for some $m, n \ge 0$. 
It will be convenient to assume that these $\alpha$-prefixes are sorted in the ascending order of their length. 
Likewise, we assume that $|\overline{v_1}| < |\overline{v_2}| < \cdots < |\overline{v_n}|$. 

Our investigation on the properties of $\alpha$-prefix and $\calpha$-suffix of word begins with a basic observation. 

\begin{proposition}\label{prop:aprefix_basic}
	For a word $w \in \alpha \Sigma^*$, the following statements hold: 
	\begin{enumerate}
	\item	for any $u \in \Pref_\alpha(w)$, $\alpha \le_p u\alpha$; 
	\item	for any $x_1, \ldots, x_n \in \Pref_\alpha(w)$, $\alpha \le_p x_1 \cdots x_n \alpha$; 
	\end{enumerate}
\end{proposition}
\begin{proof}
	The first statement derives directly from the definition of $\alpha$-prefix. 
	For the second one, induction on $n$ works. 
	Due to the first statement, $\alpha \le_p x_n \alpha$ so that proving $\alpha \le_p x_1 \cdots x_{n-1} x_n \alpha$ is reduced to proving $\alpha \le_p x_1 \cdots x_{n-1} \alpha$. 
\end{proof}

From this proposition, we can easily deduce that for a word $w \in \Sigma^* \overline{\alpha}$ and $\overline{y_1}, \ldots, \overline{y_t} \in \Suff_{\overline{\alpha}}(w)$, $\bar{\alpha} \overline{y_1} \cdots \overline{y_t} \ge_s \overline{\alpha}$, which means $\alpha \le_p y_t \cdots y_1 \alpha$. 
This deepens the above observation further as follows. 

\begin{corollary}\label{cor:apref_casuff_transitive}
	For a word $w \in \alpha \Sigma^* \cap \Sigma^* \overline{\alpha}$, any word in $(\Pref_\alpha(w) \cup \overline{\Suff_{\overline{\alpha}}(w)})^* \alpha$ has $\alpha$ as its prefix. 
\end{corollary}

Due to the second statement of Proposition~\ref{prop:aprefix_basic}, $\alpha \le_p x_1 \alpha \le_p x_1 x_2 \alpha \le_p \cdots \le_p x_1 x_2 \cdots x_s \alpha$ holds for $\alpha$-prefixes $x_1, \ldots, x_s \in \Pref_\alpha(w)$. 
Hence, from a word $x_1 x_2 \cdots x_s \alpha w' \calpha$, one-step right hairpin completion can produce at least the words $x_1 x_2 \cdots x_s \alpha w' \calpha \{\lambda, \overline{x_1}, \overline{x_1x_2}, \ldots, \overline{x_1 x_2 \cdots x_s}\}$.\footnote{$\overline{x_1x_2 \cdots x_s} = \overline{x_s} \cdots \overline{x_2} \ \overline{x_1}$.} 
Now, if we know that one-step hairpin completion extends the word to the right by $\overline{u}$, what can we say about the word $u$? 
Firstly, as long as $|u| \le |x_1 \cdots x_s|$, we can say that $u\alpha \le_p x_1 \cdots x_s \alpha$ by definition of hairpin completion. 
Moreover, Corollary~\ref{cor:apref_casuff_transitive} enables us to find $0 \le i < s$ such that $|x_1 \cdots x_i| < |u| \le |x_1 \cdots x_{i+1}|$. 
Then, one can let $u = x_1 \cdots x_i z$ for some prefix $z$ of $x_{i+1}$. 
Since $z \alpha \le_p x_{i+1} \alpha \le_p w$, $z$ is an $\alpha$-prefix of $w$ that is properly shorter than $x_{i+1}$. 
By defining $\ind(x_{i+1})$ to be the index satisfying $u_{\ind(x_{i+1})} = x_{i+1}$, we have $z \in \{u_1, \ldots, u_{\ind(x_{i+1})-1}\}$; recall that elements of $\Pref_\alpha(w)$ is sorted with respect to their length. 
The above argument is summarized by the next lemma. 

\begin{lemma}\label{lem:aprefix_prefix}
	Let $x_1, \ldots, x_s \in \Pref_\alpha(w)$. 
	If a word $u$ satisfies $u\alpha \le_p x_1 \cdots x_s \alpha$, then there exists an integer $0 \le i < s$ such that $u = x_1 \cdots x_i z$ for some $z \in \{u_1, \ldots, u_{\ind(x_{i+1})-1}\}$. 
\end{lemma}

A more natural setting is to assume that each of $x_1, \ldots, x_s$ is either an element of $\Pref_\alpha(w)$ {\it or} an element of $\overline{\Suff_{\calpha}(w)}$ because, by left hairpin completion, the complement of a $\calpha$-suffix of $w$ can be produced to the left of $w$. 
We need to generalize the function $\ind$ by extending its domain as follows: for $x_i \in \overline{\Suff_{\calpha}(w)}$, $\ind(x_i) = j$ if $x_i = v_j$. 
Note that this generalized $\ind$ is not a function any more in cases when $\Pref_\alpha(w) \cap \overline{\Suff_{\calpha}(w)} \neq \emptyset$, but this will not cause any problem in this paper. 

\begin{lemma}\label{lem:apref_casuf_apref}
	Let $x_1, \ldots, x_s \in \Pref_\alpha(w) \cup \overline{\Suff_{\calpha}(w)}$. 
	If a word $u$ satisfies $u\alpha \le_p x_1 \cdots x_t \alpha$, then there exists an integer $0 \le i < s$ such that $u = x_1 \cdots x_i z$, where 
	\[
	\begin{cases}
	z \in \{u_1, \ldots, u_{\ind(x_{i+1})-1}\} & \text{if $x_{i+1} \in \Pref_\alpha(w)$}; \\
	z \in \{v_1, \ldots, v_{\ind(x_{i+1})-1}\} & \text{if $x_{i+1} \in \overline{\Suff_{\calpha}(w)}$}. 
	\end{cases}
	\] 
\end{lemma}
\begin{proof}
	As done previously, we can find $0 \le i < s$ and a nonempty word $z \in \Sigma^+$ satisfying $u = x_1 \cdots x_i z$ and $z \alpha \le_p x_{i+1} \alpha$. 
	Since this prefix relation can be rewritten as $\bar{\alpha} \overline{x_{i+1}} \ge_s \bar{\alpha} \overline{z}$, if $\overline{x_{i+1}}$ is an $\overline{\alpha}$-suffix of $w$, so is $\overline{z}$. 
	The case when $x_{i+1} \in \Pref_\alpha(w)$ is clear from the previous argument. 
\end{proof}

Having considered prefix relations among $\alpha$-prefixes and $\calpha$-suffixes of a word, now we proceed our study to more general factor relationships among them. 

\begin{lemma}\label{lem:suf_relation_aprefs}
	If $u_j \alpha \ge_s u_i \alpha$ for some integers $2 \le i \le j \le m$, then $u_j \in \{u_1, u_2, \ldots, u_{j-1}\}u_i$. 
\end{lemma}
\begin{proof}
	We can let $x u_i \alpha = u_j \alpha$ for some $x \in \Sigma^*$. 
	Combining this with Proposition~\ref{prop:aprefix_basic}, we have $x \alpha \le_p u_j \alpha$ so that $x \in \Pref_\alpha(w)$. 
	Since $|x| < |u_j|$, $x$ is in $\{u_1, u_2, \ldots, u_{j-1}\}$. 
\end{proof}

\begin{lemma}\label{lem:factor_relation_shortests}
	If $v_2 \alpha$ is a factor of $u_2 \alpha$, then $u_2 = v_2$. 
\end{lemma}
\begin{proof}
	Let $u_2 \alpha = xv_2 \alpha y$ for some $x, y \in \Sigma^*$. 
	Unless $y = \lambda$, $xv_2 \alpha \le_p u_2 \alpha$ would be a nonempty $\alpha$-prefix of $w$ that is properly shorter than $u_2$, and causes a contradiction. 
	Thus, $y$ must be empty so that $u_2 \alpha = x v_2 \alpha$. 
	Now, Lemma~\ref{lem:suf_relation_aprefs} leads us to $x = \lambda$. 
\end{proof}

Finally, let us introduce interesting results that illustrate the close relationship between $\alpha$-prefixes, commutativity, and primitivity, essential notions in combinatorics on words. 

\begin{lemma}\label{lem:primitive_apref}
	Let $w \in \alpha \Sigma^*$ and $u \in \Pref_\alpha(w)$. 
	Then $\rho(u), \rho(u)^2, \ldots, \rho(u)^{|u|/|\rho(u)|} \in \Pref_\alpha(w)$. 
\end{lemma}
\begin{proof}
	Due to the first statement of Proposition~\ref{prop:aprefix_basic}, $u \in \Pref_\alpha(w)$ enables us to let $\alpha y = u \alpha$ for some $y \in \Sigma^+$. 
	Its solution is well-known to be $u = (st)^n$ and $\alpha = (st)^i s$ for some $i \ge 0$ and $s, t \in \Sigma^*$ such that $\rho(u) = st$. 
	Hence, $u \alpha = (st)^{i+n}s = \rho(u) \alpha (ts)^{n-1} = \rho(u)^2 \alpha (ts)^{n-2} = \cdots = \rho(u)^n \alpha$. 
\end{proof}

An immediate implication of this lemma is that the shortest nonempty $\alpha$-prefix of a word that begins with $\alpha$ must be primitive. 
We should make one more step forward. 
Imagine that a word $w$ has an $\alpha$-prefix $u$. 
If $w \to_{\RHC} w \overline{u}$ is possible, then $w \to_{\RHC} w\overline{\rho(u)}$ is also possible. 
Thus, repeating the extension of $w$ to the right by $\overline{\rho(u)}$ $|u|/|\rho(u)|$ times amounts to extending $w$ by $\overline{u}$ once. 
In other words, the process to extend a word by $\overline{u}$ is not essential unless $u$ is primitive because it can be always simulated by multiple processes to extend a word by $\overline{\rho(u)}$. 

The next lemma proves that all nonempty $\alpha$-prefixes of length at most $|\alpha|$ commute with each other, and hence, only the shortest one is essential in the above sense. 

\begin{lemma}
	For nonempty words $x_1, x_2 \in \Sigma^+$, if $\alpha \le_p x_1 \alpha \le_p x_2 \alpha$ and $|x_2| \le |\alpha|$ hold, then $\rho(x_1) = \rho(x_2)$. 
\end{lemma}
\begin{proof}
	If $|x_1| = |x_2|$, then the prefix relation immediately gives $x_1 = x_2$, and the conclusion of this lemma is trivial. 
	Hence, we assume $|x_1| < |x_2|$. 
	Combining $|x_1| \le |\alpha|$ with $\alpha \le_p x_1 \alpha$, we can deduce that the word $x_1 \alpha$ has a period $|x_1|$. 
	Likewise, $x_2\alpha$ has a period $|x_2|$, and hence, $x_1 \alpha$ also has this period. 
	As a result, $x_1\alpha$ has two periods $|x_1|, |x_2|$, and moreover it is of length at least the sum of these periods. 
	Thus, Fine and Wilf's theorem \cite{ChKa97, FiWi65} leads us to the conclusion of this lemma. 
\end{proof}

	\subsection{Non-crossing words and their properties}

A word $w_0 \in \Sigma^*$ is an {\it $(m, n)$-$\alpha$-word}, or simply an {\it $(m, n)$-word} when $\alpha$ is clear from the context, if $|\Pref_\alpha(w_0)| = m$ and $|\Suff_{\overline{\alpha}}(w_0)| = n$. 
Informally speaking, an $(m, n)$-word is a word on which $\alpha$ occurs $m$ times and $\overline{\alpha}$ does $n$ times. 
For a pseudo-palindromic $\alpha$ ($\alpha = \overline{\alpha}$), we regard an occurrence of $\alpha$ also as that of $\overline{\alpha}$, and as such, any word is an $(m, m)$-word for some $m \ge 0$. 

We say that $w_0$ is {\it non-$\alpha$-crossing} if the rightmost occurrence of $\alpha$ precedes the leftmost one of $\overline{\alpha}$ on $w_0$. 
When $\alpha$ is understood from the context, we simply say that $w_0$ is non-crossing. 
Otherwise, the word is {\it $\alpha$-crossing} or {\it crossing}. 
Note that if $\alpha = \overline{\alpha}$, then for a word $w$ which is either a $(0, 0)$-word or $(1, 1)$-word, $\HC_\alpha^*(w) = \{w\}$, and otherwise ($w$ is an $(m, m)$-word for some $m \ge 2$), $w$ can be considered crossing. 
Thus, whenever the non-$\alpha$-crossing word is concerned, we assume that $\alpha \neq \overline{\alpha}$. 
The definition of a word being non-$\alpha$-crossing does not force the word to begin with $\alpha$ or end with $\calpha$. 
However, it is not until $\alpha$ is a primer that this notion becomes useful in our work. 
Thus, the word should be in either $\alpha \Sigma^*$ or $\Sigma^* \calpha$. 
Actually, in the rest of this paper, we assume both of these conditions and consider only {\it single-primer iterated hairpin completion}; thus, we can assume that $w_0 \in \alpha \Sigma^* \cap \Sigma^* \calpha$. 
As let previously, elements of $\Pref_\alpha(w_0)$ are denoted by $u_1, \ldots, u_m$, those of $\Suff_{\calpha}(w_0)$ by $\overline{v_1}, \ldots, \overline{v_n}$, and they are sorted so that this assumption imposes $u_1 = v_1 = \lambda$. 

Our main focus lies on the characterization of non-crossing words whose iterated hairpin completion is regular in terms of combinatorics on words. 
Thus, in this subsection, we prove some combinatorial properties of non-crossing words. 
Let us begin with an easy observation about the longest $\alpha$-prefix and $\calpha$-suffix of $w_0$. 

\begin{proposition}\label{prop:mirror}
	$u_m = v_n$ if and only if $m = n$ and for all $1 \le i \le m$, $u_i = v_i$. 
\end{proposition}

Next, we will see that one-step hairpin completion can extend $w_0$ to the left by any of $v_1, \ldots, v_{n-1}$ or to the right by any of $\overline{u_1}, \ldots, \overline{u_{m-1}}$ due to the following lemma. 

\begin{lemma}\label{lem:length_nonoverlap_apref_casuf}
	Let $w_0 \in \alpha \Sigma^* \cap \Sigma^* \overline{\alpha}$ be a non-crossing word with $\Pref_\alpha(w_0) = \{u_1, \ldots, u_m\}$ and $\Suff_{\overline{\alpha}}(w_0) = \{\overline{v_1}, \ldots, \overline{v_n}\}$. 
	Then $|u_{m-1}| + |v_n| + 2|\alpha| < |w_0|$. 
\end{lemma}
\begin{proof}
	Suppose that this inequality did not hold. 
	Being non-crossing, $w_0$ can be written as $w_0 = u_{m-1} w \overline{v_n}$ for some $w \in \alpha \Sigma^* \cap \Sigma^* \overline{\alpha}$ with $|w| \le 2|\alpha|$. 
	Hence, $w = \overline{w}$. 
	Let $x$ be a nonempty word satisfying $u_m = u_{m-1}x$. 
	Since $w_0$ is non-crossing, $u_m \alpha \le_p u_{m-1} w$ must hold, from which we have $x \alpha \le_p w$. 
	Combining this with $w = \overline{w}$ enables us to find an $\calpha$-suffix $\overline{x} \ \overline{v_n}$ of $w_0$, but this would be longer than the longest $\overline{\alpha}$-suffix of $w_0$, a contradiction. 
\end{proof}

This lemma does not rule out the possibility that $w_0$ cannot be extended to the right by $\overline{u_m}$ by hairpin completion because the rightmost occurrence of $\alpha$ might overlap with the suffix $\calpha$. 
The analogous argument is valid for $v_n$ and left hairpin completion. 
However, Lemma~\ref{lem:length_nonoverlap_apref_casuf} leads us to one important corollary on non-crossing $(m, n)$-words for $m, n \ge 2$ that hairpin completion can extend $w_0$ to the right by the complement of any of its $\alpha$-prefix and to the left by the complement of any of its $\calpha$-suffix. 

\begin{corollary}\label{cor:1st_step_mn_ge2}
	Let $w_0 \in \alpha \Sigma^* \cap \Sigma^* \overline{\alpha}$ be a non-crossing $(m, n)$-word with $m, n \ge 2$. 
	Then $\HC_\alpha(w_0) = \{w_0\} \cup \{v_2, \ldots, v_m\}w_0 \cup w_0\{\overline{u_2}, \ldots, \overline{u_n}\}$. 
\end{corollary}

Any word obtained from a non-crossing word by hairpin completion is non-crossing. 
Though being easily confirmed, this closure property forms the foundation of our discussions in this paper. 

\begin{proposition}\label{prop:nonx_HC}
	Let $\alpha \in \Sigma^k$ with $\alpha \neq \overline{\alpha}$, and $w_0 \in \alpha \Sigma^* \cap \Sigma^* \overline{\alpha}$ be a non-crossing word. 
	Then any word in $\HC_\alpha^*(w_0)$ is non-crossing. 
\end{proposition}

We conclude this section with a characterization of a non-$\alpha$-crossing word in terms of minimal factors with respect to the language $\alpha \Sigma^* \cap \Sigma^* \calpha$. 
With Proposition~\ref{prop:nonx_HC}, this characterization will bring a unique factorization theorem (Theorem~\ref{thm:nonx_initial_once}) of any word $w$ in $\HC_\alpha^*(w_0)$ as $w = xw_0y$ for some words $x, y$. 

\begin{lemma}
	Let $\alpha \in \Sigma^k$ with $\alpha \neq \calpha$. 
	A word $w_0 \in \alpha\Sigma^* \cap \Sigma^*\calpha$ is non-crossing if and only if it contains exactly one minimal factor $v$ from $\alpha\Sigma^* \cap \Sigma^*\overline\alpha$.
\end{lemma}
\begin{proof}
	Let us consider the contrapositive of the converse implication. 
	So, if $w_0$ is crossing, then we can find an occurrence of $\calpha$ (let us denote it by $\calpha_0$) which precedes an occurrence of $\alpha$ ($\alpha_1$). 
	$\calpha_0$ is guaranteed to be preceded by another occurrence of $\alpha$ ($\alpha_2$) because $w_0$ begins with $\alpha$. 
	Thus, the factor of $w_0$ that spans from $\alpha_2$ to $\calpha_0$ is a minimal factor from $\alpha \Sigma^* \cap \Sigma^* \calpha$. 
	By the same token, the factor of $w_0$ that spans from $\alpha_2$ to its right adjacent occurrence of $\calpha$ becomes another minimal factor. 

	In order to prove the direct implication, suppose that $w_0$ contains two minimal factors from $\alpha \Sigma^* \cap \Sigma^* \calpha$. 
	These two factors must overlap with each other because otherwise the suffix $\calpha$ of the first factor precedes the prefix $\alpha$ of the second one and $w$ would be crossing. 
	However, if they overlap, then the overlapped part would be in $\alpha \Sigma^* \cap \Sigma^* \calpha$, and this contradicts the minimality of the two factors. 
\end{proof}

\begin{theorem}\label{thm:nonx_initial_once}
	Let $\alpha \in \Sigma^k$ with $\alpha \neq \overline{\alpha}$, and $w_0 \in \alpha \Sigma^* \cap \Sigma^* \overline{\alpha}$ be a non-crossing word. 
	On any word in $\HC_\alpha^*(w_0)$, $w_0$ occurs exactly once as a factor. 
\end{theorem}
\begin{proof}
	From the two facts that any word in $\HC_\alpha^*(w_0)$ is non-crossing (Proposition~\ref{prop:nonx_HC}) and that these words contain at least one occurrence of $w_0$ as a factor by definition of hairpin completion, we can reach this conclusion. 
\end{proof}

	\section{Iterated hairpin completion of non-crossing words}

This section contains the main contribution of this paper: characterizations of the regularity of iterated hairpin completion of a non-crossing $(m, n)$-word $w_0 \in \alpha \Sigma^* \cap \Sigma^* \calpha$ (recall that $\alpha \neq \calpha$ is assumed). 
Throughout this section, $w_0$ is thus assumed with $\Pref_\alpha(w_0) = \{u_1, \ldots, u_m\}$ and $\Suff_{\calpha}(w_0) = \{\overline{v_1}, \ldots, \overline{v_m}\}$. 

Let us begin with a proof that one-sided hairpin completion of a non-crossing word is regular (Theorem~\ref{thm:nonx_oneside_regular}). 
Then we will show that the iterated hairpin completion of a non-crossing $(m, 1)$-word for any $m \ge 1$ or $(2, 2)$-word is always regular (Theorems~\ref{thm:m1nonx_regular} and \ref{thm:22nonx_regular}). 
Using these results and combinatorial results shown in Section~\ref{sec:apref_casuf}, we characterize the set of all non-crossing $(3, 2)$-words whose iterated hairpin completion is regular, in terms of commutativity (Theorem~\ref{thm:32nonx_iff_regular}). 

\begin{theorem}\label{thm:nonx_oneside_regular}
	For a non-crossing word $w_0 \in \alpha \Sigma^* \cap \Sigma^* \overline{\alpha}$, both $\LHC_\alpha^*(w_0)$ and $\RHC_\alpha^*(w_0)$ are regular. 
\end{theorem}
\begin{proof}
	First, we prove the regularity of $\RHC_\alpha^*(w_0)$. 
	Let $w$ be an $\alpha$-prefix of $w_0$. 
	A right hairpin completion of $w_0$ can produce $w_0\overline{w}$. 
	Note that the suffix $\bar{\alpha}\overline{w}$ of this resulting word does not contain $\alpha$ due to the non-crossing assumption on $v$, and this means that the longest $\alpha$-prefix of $w_0\overline{w}$ is the same as that of $w_0$. 
	Thus, the language $\RHC_\alpha^*(w_0)$ can be obtained by iterated bounded hairpin completion from $v$, and hence, is regular \cite{Kopecki10}. 

	For the regularity of $\LHC_\alpha^*(w_0)$, it suffices to observe that $\overline{w_0}$ is also non-crossing. 
	Using the result just proved, $\RHC_\alpha^*(\overline{w_0})$ is regular, and according to Proposition~\ref{prop:presentation_RH}, $\LHC_\alpha^*(w_0) = \overline{\RHC_\alpha^*(\overline{w_0})}$. 
	Note that the class of regular languages is closed under $\bar{\hspace{2mm}}$. 
\end{proof}

	\subsection{Iterated hairpin completion of $(m, 1)$ non-crossing words}

In this subsection, we consider the case $n = 1$ ($w_0$ is an $(m, 1)$-word), and prove that $\HC_\alpha^*(w_0)$ is regular. 
For $m = 1$, it is easy to see that hairpin completion cannot generate any word but $w_0$, that is, $\HC_\alpha^*(w_0) = \{w_0\}$. 
Hence, we assume $m \ge 2$. 

Lemma~\ref{lem:length_nonoverlap_apref_casuf} means that right hairpin completion can extend $w_0$ to the right by any of $\overline{u_1}, \overline{u_2}, \ldots, \overline{u_{m-1}}$, 
In contrast, the operation can extend $w_0$ to the right by $\overline{u_m}$ if and only if $|u_m| + 2|\alpha| + |v_1| \ge |w_0|$, i.e., the $\alpha$ to the right of $u_m$ does not overlap with the suffix $\calpha$ of $w_0$. 
As a result, if $m = 2$ but this inequality does not hold, then $\HC_\alpha^*(w_0) = \{w_0\}$. 
Therefore, we can advance our discussion on the assumption that $w_0 \to_{\RHC} w_0 \overline{u_2}$ is valid. 

Note that $w_0 \overline{u_2}$ is a non-crossing $(m, 2)$-word. 
Applying Lemma~\ref{lem:length_nonoverlap_apref_casuf} to this word, we can see that $|u_m| + 2|\alpha| < |w_0 \overline{u_2}|$. 
Hence, hairpin completion can extend $w_0 \overline{u_2}$ further to the right by not only by any of $\overline{u_1}, \overline{u_2}, \ldots, \overline{u_{m-1}}$ but also by $\overline{u_m}$. 

Let us define the following regular language: 
\[
\begin{array}{ll}
	R_{m1}(w_0) = \{w_0\} \cup \bigl\{x_s \cdots x_1 w_0 \overline{y_1} \ \overline{y_2} \cdots \overline{y_t} \bigm|
		& y_1 \in \begin{cases}
			\{u_1, \ldots, u_{m-1}, u_m\} 	& \text{if $|u_m| + 2|\alpha| \le |w_0|$} \\
			\{u_1, \ldots, u_{m-1}\}	& \text{otherwise} 
			\end{cases} \\
		& s \ge 0, t \ge 1, x_s, \ldots, x_1, y_2, \ldots, y_t \in \{u_1, \ldots, u_m\}, \\
		& \mbox{and } \max_{1 \le i \le s}\{\ind(x_i)\} \le \max_{1 \le j \le t}\{\ind(y_j)\} \bigr\}. \\
\end{array}
\]
We claim that this language is the language obtained from $w_0$ by iterated hairpin completion. 

First, we prove that $\HC_\alpha^*(w_0) \supseteq R_{m1}(w_0)$. 
Let $w \in R_{m1}(w_0)$. 
By definition, any word in $R_{m1}(w_0)$ can be factorized as $w = x_s \cdots x_1 w_0 \overline{y_1} \ \overline{y_2} \cdots \overline{y_t}$. 
Compare the leftmost factor $x_s$ and the complement of the rightmost factor $\overline{y_t}$ with respect to their index. 
Assume that $\ind(x_s) \le \ind(y_t)$. 
Then $w \ge_s \bar{\alpha} \overline{y_t} \ge_s \bar{\alpha} \overline{x_s}$. 
Hence, one-step left hairpin completion can derive $w$ from the word $x_{s-1} \cdots x_1 w_0 \overline{y_1} \cdots \overline{y_t}$. 
In the case when $\ind(x_1) > \ind(y_t)$, the same argument implies that $w \in \RHC_\alpha(x_s \cdots x_1 w_0 \overline{y_1} \cdots \overline{y_{t-1}})$. 
Due to $\max_{1 \le i \le s}\{\ind(x_i)\} \le \max_{1 \le j \le t}\{\ind(y_j)\}$, the repetition of this process eventually reduces $w_0$ into a word $w_0 \overline{y_1} \cdots \overline{y_j}$ for some $1 \le j \le t$. 
Because of the condition on $y_1$ and our discussion above, $w_0 \to_{\RHC} w_0 \overline{y_1} \to_{\RHC} \cdots \to_{\RHC} w_0 \overline{y_1} \cdots \overline{y_j}$ is valid. 
Thus, $w \in \HC_\alpha^*(w_0)$. 

Secondly, we prove the opposite inclusion by induction on the length of derivation by hairpin completion. 
Clearly $w_0 \in L(w_0)$. 
Let us assume that a word in $\HC_\alpha^*(w_0)$ can be written as $x_s \cdots x_1 w_0 y_1 \cdots y_t$ with $\max_{1 \le i \le s}\{\ind(x_i)\} \le \max_{1 \le j \le t}\{\ind(y_j)\}$. 
Let $j = \max_{1 \le j \le t}\{\ind(y_j)\}$. 
If left hairpin completion extends this word to the left by $x$, then $\bar{\alpha} \overline{y_1} \cdots \overline{y_t} \ge_s \bar{\alpha}\bar{x}$ and this means $x \in \{u_1, \ldots, u_j\}^+$ (see Lemma~\ref{lem:aprefix_prefix}). 
Thus, there exist $x_{s'}, \ldots, x_{s+1} \in \{u_1, \ldots, u_j\}$ such that $x = x_{s'} \cdots x_{s+1}$ and $\max\{\ind(x_{s'}), \ldots, \ind(x_{s+1}), \ind(x_s), \ldots, \ind(x_1)\} \le j$. 
It it trivial that this inequality remains valid in the right hairpin completion.

\begin{theorem}\label{thm:m1nonx_regular}
	For any $m \ge 1$ and a non-crossing $(m, 1)$ word $w_0 \in \alpha \Sigma^* \overline{\alpha}$, the language $\HC_\alpha^*(w_0)$ is regular. 
\end{theorem}

The key idea in the above discussion is that if a word in $\HC_\alpha^*(w_0)$ begins with the longest $\alpha$-prefix $u_m$ of $w_0$, then hairpin completion can extend it to the right by any of $\alpha$-prefix of $w_0$. 
This idea has a broader range of applications. 
Let $w_0 \in \alpha \Sigma^* \cap \Sigma^* \calpha$ be a non-crossing $(m, n)$-word for some $m, n \ge 1$ with $\Pref_\alpha(w_0) = \{u_1, \ldots, u_m\}$ and $\Suff_{\calpha}^*(w_0) = \{\overline{v_1}, \ldots, \overline{v_n}\}$. 
Proposition~\ref{prop:mirror} says that if $u_m = v_n$, then $\Suff_{\calpha}^*(w_0) = \overline{\Pref_\alpha^*(w_0)}$. 
For $m \ge 2$, the rightmost occurrence of $\alpha$ on $w_0$ does not overlap with the suffix $\calpha$ of $w_0$ (Lemma~\ref{lem:length_nonoverlap_apref_casuf}). 
Thus, $\HC_\alpha^*(w_0) = \{u_1, \ldots, u_m\}^* w_0 \{\overline{u_1}, \ldots, \overline{u_m}\}^*$. 

\begin{corollary}\label{cor:mn_palindrome-like_regular}
	Let $w_0 \in \alpha \Sigma^* \cap \Sigma^* \calpha$ be a non-crossing $(m, n)$-word. 
	If $u_m = v_n$, then $\HC_\alpha^*(w_0)$ is regular. 
\end{corollary}

	\subsection{Iterated hairpin completion of $(2, 2)$ non-crossing words}

In contrast to the result obtained in the previous subsection, Example~\ref{ex:nonx_cs} shows that there exists an $(m, 2)$ non-crossing word whose iterated hairpin completion is non-regular with $m = 3$. 
This result motivates the study of $(2, 2)$ non-crossing words reported here. 
Let $w_0 \in \alpha \Sigma^* \cap \Sigma^* \calpha$ be a non-crossing $(2, 2)$-word. 
We can employ Corollary~\ref{cor:1st_step_mn_ge2} to see that $\HC_\alpha(w_0) = \{w_0, v_2 w_0, w_0 \overline{u_2}\}$. 
This further implies that the suffix $\calpha$ of any word in $\HC_\alpha^*(w_0)$ can bind with the second $\alpha$ on the (unique) factor $w_0$ on the word for right hairpin completion. 

Let us define the following regular language: 
\[
	R_{22L} = v_2^* (v_2 w_0) \overline{v_2}^* \cup (v_2^+ u_2)^* v_2^* (v_2 w_0) \overline{v_2}^* (\overline{u_2} \ \overline{v_2}^+)^+.
\]
We will show that this language is exactly the set of words obtained by iterated hairpin completion from $v_2 w_0$. 

In order to prove that $\HC_\alpha^*(v_2 w_0) \supseteq R_{22L}$, it suffices to present the following process: 
\begin{eqnarray*}
	v_2 w_0 &\to_{\RHC}^*& 	v_2 w_0 \overline{v_2}^{j_0} \\
		&\to_{\RHC}& 	v_2 w_0 \overline{v_2}^{j_0} \overline{u_2} \ \overline{v_2} \\
		&\to_{\RHC}^*& 	v_2 w_0 \overline{v_2}^{j_0} \overline{u_2} \ \overline{v_2}^{j_1} \\ 
		&\to_{\RHC}& 	v_2 w_0 \overline{v_2}^{j_0} \overline{u_2} \ \overline{v_2}^{j_1} \overline{u_2} \ \overline{v_2} \\
		&\to_{\RHC}^*& 	v_2 w_0 \overline{v_2}^{j_0} \overline{u_2} \ \overline{v_2}^{j_1} \cdots \overline{u_2} \ \overline{v_2}^{j_{t-1}} \overline{u_2} \ \overline{v_2} \\
		&\to_{\LHC}^*& 	v_2^{i_0} v_2 w_0 \overline{v_2}^{j_0} \overline{u_2} \ \overline{v_2}^{j_1} \cdots \overline{u_2} \ \overline{v_2}^{j_{t-1}} \overline{u_2} \ \overline{v_2} \\ 
		&\to_{\LHC}& 	v_2 u_2 v_2^{i_0} v_2 w_0 \overline{v_2}^{j_0} \overline{u_2} \ \overline{v_2}^{j_1} \cdots \overline{u_2} \ \overline{v_2}^{j_{t-1}} \overline{u_2} \ \overline{v_2} \\
		&\to_{\LHC}^*& 	v_2^{i_1} u_2 v_2^{i_0} v_2 w_0 \overline{v_2}^{j_0} \overline{u_2} \ \overline{v_2}^{j_1} \cdots \overline{u_2} \ \overline{v_2}^{j_{t-1}} \overline{u_2} \ \overline{v_2} \\
		&\to_{\LHC}^*& 	v_2^{i_s} u_2 \cdots v_2^{i_1} u_2 v_2^{i_0} v_2 w_0 \overline{v_2}^{j_0} \overline{u_2} \ \overline{v_2}^{j_1} \cdots \overline{u_2} \ \overline{v_2}^{j_{t-1}} \overline{u_2} \ \overline{v_2} \\
		&\to_{\RHC}^*& 	v_2^{i_s} u_2 \cdots v_2^{i_1} u_2 v_2^{i_0} v_2 w_0 \overline{v_2}^{j_0} \overline{u_2} \ \overline{v_2}^{j_1} \cdots \overline{u_2} \ \overline{v_2}^{j_{t-1}} \overline{u_2} \ \overline{v_2}^{j_t}.
\end{eqnarray*}

Next, we prove the opposite inclusion by induction on the length of derivation by hairpin completion from $v_2 w_0$. 
Obviously, $v_2 w_0 \subseteq R_{22L}$. 
Assume that all words obtained from $v_2 w_0$ by at most $n$-times hairpin completion are in $R_{22L}$. 
Let $w_n$ be such a word and consider a word $w_{n+1}$ such that $w_n \to_{\HC}~w_{n+1}$. 
Consider the case when this hairpin completion is right one. 
The rightmost occurrence of $\alpha$ on $w_n$ is the second $\alpha$ on its (unique) factor $w_0$. 
Therefore, if we let $w_{n+1} = w_n \overline{x}$ and then $x \alpha \le_p (v_2^+ u_2)^*v_2^*v_2u_2$. 
Since $u_2$ and $\overline{v_2}$ are the respective shortest nonempty $\alpha$-prefix and $\calpha$-suffix of $w_0$, Lemma~\ref{lem:apref_casuf_apref} implies that $x \in (v_2^+ u_2)^* v_2^*$. 
Note that $R_{22L}$ is closed under catenating a word in $\overline{(v_2^+ u_2)^* v_2^*}$ to the right. 
Thus, $w_{n+1} \in R_{22L}$. 
The case when $w_n \to_{\LHC} w_{n+1}$ can be proved in a symmetric manner. 

Due to the symmetry of $u_2$ and $\overline{v_2}$, we can easily construct a regular language $R_{22R}$ which is equivalent to $H_\alpha^*(w_0 \overline{u_2})$. 
Now the regularity of $\HC_\alpha^*(w_0)$ has been proved.

\begin{theorem}\label{thm:22nonx_regular}
	For a $(2, 2)$ non-crossing word $w_0 \in \alpha \Sigma^* \overline{\alpha}$, the language $\HC_\alpha^*(w_0)$ is regular. 
\end{theorem}

	\subsection{Iterated hairpin completion of $(3, 2)$ non-crossing words}

Theorem~\ref{thm:22nonx_regular} and Example~\ref{ex:nonx_cs} motivate our investigation of non-crossing $(3, 2)$ words. 
Actually, Theorem~\ref{thm:32nonx_iff_regular}, a main contribution of this paper, provides a characterization of the regularity of iterated hairpin completion of a non-crossing $(3, 2)$-word in terms of the commutativity of the $\alpha$-prefixes and $\calpha$-suffixes of the word. 

Let $w_0 \in \alpha \Sigma^* \cap \Sigma^* \overline{\alpha}$ be a non-crossing $(3, 2)$-word (so $\alpha \neq \overline{\alpha}$) with $\Pref_\alpha(w_0) = \{\lambda, u_2, u_3\}$ and $\Suff_{\calpha}(w_0) = \{\lambda, \overline{v_2}\}$. 
Note that $u_2$ ($v_2$) must be primitive; otherwise, its primitive root is also an $\alpha$-prefix (resp.~$\calpha$-suffix) of $w_0$ and $w_0$ would not be a $(3, 2)$-word any more. 
As a result, $u_2$ commute with $v_2$ ($u_3$) if and only if $u_2 = v_2$ (resp.~$u_3 = u_2^2$). 
Recall also that $u_3 \neq v_2$ must hold for $w_0$ to be $(3, 2)$-word (Proposition~\ref{prop:mirror}).
Thus, if $u_3$ and $v_2$ commute, then $u_3 = v_2^2$ and $u_2 = v_2$. 
In other words, the commutativity between $u_3$ and $v_2$ is reduced to the commutativity between $u_2$ and $u_3$ and the commutativity between $u_2$ and $v_2$, and hence, not essential. 

Corollary~\ref{cor:1st_step_mn_ge2} states that $\HC_\alpha(w_0) = \{w_0\} \cup \{v_2 w_0, w_0 \overline{u_2}, w_0 \overline{u_3}\}$. 
Let us ask the question of whether iterated hairpin completion can generate a same word from $w_0 \overline{u_2}$ and $w_0 \overline{u_3}$. 
We partially answer this question in a broader setting for arbitrary $m \ge 3$ and $n \ge 1$. 

\begin{lemma}\label{lem:mn_nonx_include_or_emptyset}
	Let $w_0 \in \alpha \Sigma^* \cap \Sigma^* \calpha$ be a non-crossing $(m, n)$-word for some $m \ge 3$ and $n \ge 1$ with $\Pref_\alpha(w_0) = \{u_1, \ldots, u_m\}$. 
	For integers $i, j$ with $1 < i < j$, if $u_j \in \{u_2, \ldots, u_{j-1}\}u_i$, then $\HC_\alpha^*(w_0 \overline{u_j}) \subseteq \HC_\alpha^*(w_0 \overline{u_i})$; otherwise, $\HC_\alpha^*(w_0 \overline{u_j}) \cap \Sigma^* w_0 \overline{u_i} \Sigma^* = \emptyset$. 
\end{lemma}
\begin{proof}
	Let $u_j = x u_i$ for some $x \in \{u_2, \ldots, u_{j-1}\}$.
	Lemma~\ref{lem:length_nonoverlap_apref_casuf} implies that $w_0 \overline{u_i} \to_{\RHC} w_0 \overline{u_i} \ \overline{x} = w_0 \overline{u_j}$ is possible. 
	Thus, the inclusion holds. 
	Conversely, if the intersection is not empty, then Theorem~\ref{thm:nonx_initial_once} implies that $\calpha \ \overline{u_j} = \calpha \ \overline{u_i} \ \overline{y}$ for some $y \in \Sigma^+$. 
	Then, due to Lemma~\ref{lem:suf_relation_aprefs}, this equation gives $y \in \{u_2, \ldots, u_{j-1}\}$; thus, $u_j \in \{u_2, \ldots, u_{j-1}\}u_i$. 
\end{proof}

We can employ Lemma~\ref{lem:mn_nonx_include_or_emptyset} in our current setting of non-crossing $(3, 2)$-words to observe that if $u_3 = u_2^2$, then $\HC_\alpha^*(w_0 \overline{u_3}) \subseteq \HC_\alpha^*(w_0 \overline{u_2})$; otherwise, $\HC_\alpha^*(w_0 \overline{u_3}) \cap \Sigma^* w_0 \overline{u_2} \Sigma^* = \emptyset$. 
Thus, for example, if $u_3 \neq u_2^2$, then $\HC_\alpha^*(w_0\overline{u_3}) \cap \HC_\alpha^*(w_0 \overline{u_2}) = \emptyset$. 

In this subsection, we first prove that the commutativity of $u_2$ with $v_2$ or with $u_3$ is a sufficient condition for $\HC_\alpha^*(w_0)$ to be regular. 

\begin{lemma}\label{lem:22nonx_u2isv2_regular}
	If $u_2 = v_2$, then the language $\HC_\alpha^*(w_0)$ is regular. 
\end{lemma}
\begin{proof}
	Let $w_0 = w \overline{v_2}$ for some $w \in \alpha \Sigma^* \cap \Sigma^* \calpha$. 
	Observe that $w$ is a non-crossing $(3, 1)$-word with $u_2, u_3$ being its nonempty $\alpha$-prefix. 
	Lemma~\ref{lem:length_nonoverlap_apref_casuf} implies that $|u_2| + 2|\alpha| < |w|$, which means that hairpin completion can extend $w$ to the right by $\overline{u_2}$ and result in $w_0$. 
	If $|u_3| + 2|\alpha| \le |w|$, then hairpin completion can also generate $w \overline{u_3}$, but it is not essential in the following discussion whether this is possible or not. 
	Let us consider only the case when it is possible. 
	Then $\HC_\alpha^*(w)$, which is regular due to Theorem~\ref{thm:m1nonx_regular}, is $\{w\} \cup \HC_\alpha^*(w\overline{u_2}) \cup \HC_\alpha^*(w\overline{u_3})$. 
	As we have seen above, if $w\overline{u_3} \in \HC_\alpha(w)$, then either $\Sigma^* w \overline{u_2} \Sigma^* \cap \HC_\alpha^*(w\overline{u_3}) = \emptyset$ or $\HC_\alpha^*(w\overline{u_2}) \supseteq \HC_\alpha^*(w\overline{u_3})$. 
	In any case, $\HC_\alpha^*(w_0) = \HC_\alpha^*(w) \cap \Sigma^* w \overline{u_2} \Sigma^*$, and hence, is regular. 
\end{proof}

Now it is easy to see that $\HC_\alpha^*(w_0)$ is regular when $u_3$ commutes with $v_2$. 
Since $w_0$ is $(3, 2)$-word, $v_2$ must be primitive and $u_3$ is equal to either $v_2$ or $v_2^2$. 
In the former case, $u_2$ is a proper prefix of $v_2$ so that $w_0$ has $\overline{u_2}$ and would not be a $(3, 2)$-word. 
Thus, the latter must be the case. 
In this case, the prefix $v_2$ of $u_3$, which is the primitive root of $u_3$, is an $\alpha$-prefix of $w_0$ (Lemma~\ref{lem:primitive_apref}), and hence, in order for $w_0$ to be a $(3, 2)$-word, $u_2 = v_2$ must hold, and this brings the conclusion according to Lemma~\ref{lem:22nonx_u2isv2_regular}. 

\begin{lemma}
	If $u_3 = u_2^2$, then the language $\HC_\alpha^*(w_0)$ is regular. 
\end{lemma}
\begin{proof}
	Lemma~\ref{lem:22nonx_u2isv2_regular} makes it sufficient to consider the case when $u_2$ does not commute with $v_2$. 
	Since $\HC_\alpha^*(w_0) = \{w_0\} \cup \HC_\alpha^*(v_2 w_0) \cup \HC_\alpha^*(w_0 \overline{u_2}) \cup \HC_\alpha^*(w_0 \overline{u_2}^2)$ (when the reader check this, recall Lemma~\ref{lem:length_nonoverlap_apref_casuf}) and $\HC_\alpha^*(w_0\overline{u_2}) \supseteq \HC_\alpha^*(w_0\overline{u_2}^2)$, we will show the regularity of the second and third terms of this equation and that is enough for our purpose. 

	First, we prove that $\HC_\alpha^*(w_0 \overline{u_2})$ is regular. 
	Let $w_0 = u_2 w$, where $w \in \alpha \Sigma^* \cap \Sigma^* \calpha$ is a $(2, 2)$-word with $\Pref_\alpha(w) = \{\lambda, u_2\}$ and $\Suff_{\calpha}(w) = \{\lambda, \overline{v_2}\}$. 
	We can easily check that 
	\[
		\HC_\alpha^*(w) = \{w, w\overline{u_2}, w \overline{u_2}^2\} \cup \HC_\alpha^*(u_2 w \overline{u_2}) \cup \HC_\alpha^*(u_2 v_2 w \overline{u_2}) \cup \HC_\alpha^*(v_2w). 
	\]
	As done in the proof of Lemma~\ref{lem:22nonx_u2isv2_regular}, the non-commutativity between $u_2$ and $v_2$ implies that $(\HC_\alpha^*(u_2 v_2 w \overline{u_2}) \cup \HC_\alpha^*(v_2w)) \cap \Sigma^* u_2 w \Sigma^* = \emptyset$. 
	Thus, $\HC_\alpha^*(w) \cap \Sigma^* u_2 w \Sigma^* = \HC_\alpha^*(w_0 \overline{u_2})$. 
	Since $w$ is a non-crossing $(2, 2)$-word so that $\HC_\alpha^*(w)$ is regular (Theorem~\ref{thm:22nonx_regular}), and hence, so is $\HC_\alpha^*(w_0 \overline{u_2})$. 

	Next, we prove the regularity of $\HC_\alpha^*(v_2 w_0)$. 
	We can let $w_0 = w'\overline{v_2}$ for some $(3, 1)$-word $w'$. 
	This means that $v_2 w'$ is a $(4, 1)$-word with $\Pref_\alpha(v_2w') = \{\lambda, v_2, v_2u_2, v_2u_2^2\}$ and the empty $\calpha$-suffix. 
	Thus, 
	\[
		\HC_\alpha^*(v_2w') = \{v_2w'\} \cup \HC_\alpha^*(v_2w'\overline{v_2}) \cup \HC_\alpha^*(v_2w'\overline{u_2} \ \overline{v_2}) \cup \HC_\alpha^*(v_2w'\overline{u_2}^2 \overline{v_2}). 
	\]
	Using the essentially same argument as above, we obtain $\HC_\alpha^*(v_2w') \cap \Sigma^* v_2 w' \overline{v_2} \Sigma^* = \HC_\alpha^*(v_2w_0)$. 
	Since the iterated hairpin completion of non-crossing $(4, 1)$-word is regular (Theorem~\ref{thm:m1nonx_regular}), $\HC_\alpha^*(v_2w')$ is regular and so is $\HC_\alpha^*(v_2w_0)$. 

	Combining what have been proved in the previous two paragraphs together, we conclude the regularity of $\HC_\alpha^*(w_0)$. 
\end{proof}

To summarize the results obtained so far, any of two of the $\alpha$-prefixes and the complements of $\calpha$-suffixes of $w_0$, i.e., $u_2, u_3, v_2$, must not commute in order for $\HC_\alpha^*(w_0)$ not to be regular. 

\begin{lemma}
	If $u_3 = u_2 v_2$, then the language $\HC_\alpha^*(w_0)$ is regular. 
\end{lemma}
\begin{proof}
	Due to Lemma~\ref{lem:22nonx_u2isv2_regular}, it suffices to consider this problem under the assumption $u_2 \neq v_2$, which is equivalent to that $u_2$ does not commute with $v_2$ under our problem setting.  

	We have $\HC_\alpha^*(w_0) = \{w_0\} \cup \HC_\alpha^*(v_2 w_0) \cup \HC_\alpha^*(w_0 \overline{u_2}) \cup \HC_\alpha^*(w_0 \overline{v_2} \ \overline{u_2})$. 
	As done before, we will check that the second, third, and fourth terms of the union above are regular. 
	The regularity of the third one is from $\Pref_\alpha(w_0 \overline{u_2}) = \{\lambda, u_2, u_2 v_2\}$ and $\Suff_{\calpha}(w_0 \overline{u_2}) = \{\lambda, \overline{u_2}, \overline{v_2} \ \overline{u_2}\}$ and Corollary~\ref{cor:mn_palindrome-like_regular}. 

	In order to check that the second term is regular, let $w_0 = w_1 \overline{v_2}$, where $w_1$ is a $(3,1)$-word. 
	Then $v_2w_1$ is a $(4, 1)$-word, and 
	\[
		\HC_\alpha^*(v_2 w_1) = \{v_2 w_1\} \cup \HC_\alpha^*(v_2w_1 \overline{v_2}) \cup \HC_\alpha^*(v_2 w_1 \overline{u_2} \ \overline{v_2}) \cup \HC_\alpha^*(v_2 w_1 \overline{v_2} \ \overline{u_2} \ \overline{v_2}). 
	\]
	Since $v_2 w_1 \overline{v_2} \to_{\RHC} v_2 w_1 \overline{v_2} \ \overline{u_2} \ \overline{v_2}$ and $\HC_\alpha^*(v_2 w_1 \overline{u_2} \ \overline{v_2}) \cap \Sigma^* v_2 w_1 \overline{v_2} \Sigma^* = \emptyset$, we have $\HC_\alpha^*(v_2 w_0) = \HC_\alpha(v_2 w_1 \overline{v_2}) = \HC_\alpha^*(v_2 w_1) \cap \Sigma^* v_2 w_1 \overline{v_2} \Sigma^*$. 
	The regularity of $\HC_\alpha^*(v_2 w_1)$ is due to Theorem~\ref{thm:m1nonx_regular} so that $\HC_\alpha(v_2 w_0)$ is regular. 

	What remains to be considered is the fourth term. 
	One can let $w_0 \overline{v_2} \ \overline{u_2} = u_2 v_2 w_2$ for some non-crossing $(1, 4)$-word $w_2$. 
	Then $\HC_\alpha^*(w_2) = \{w_2\} \cup \HC_\alpha^*(u_2 w_2) \cup \HC_\alpha^*(u_2 v_2 w_2) \cup \HC_\alpha^*(u_2 v_2^2 w_2)$ holds, and we can easily see that $\HC_\alpha^*(u_2v_2w_2) = \HC_\alpha^*(w_2) \cap \Sigma^* u_2 v_2 w_2 \Sigma^*$. 
	The regularity of $\HC_\alpha^*(w_0 \overline{v_2} \ \overline{u_2})$ was proved. 
\end{proof}

\begin{theorem}\label{thm:32nonx_iff_regular}
	Let $w_0 \in \alpha \Sigma^* \cap \Sigma^* \calpha$ be a non-crossing $(3, 2)$-word with $\Pref_\alpha(w_0) = \{\lambda, u_2, u_3\}$ and $\Suff_{\calpha}(w_0) = \{\lambda, \overline{v_2}\}$. 
	Then $\HC_\alpha^*(w_0)$ is regular if and only if one of the following three conditions holds: 
	\begin{enumerate}
	\item	$u_2$ commutes with $v_2$; 
	\item	$u_2$ commutes with $u_3$; 
	\item	$u_3 = u_2 v_2$. 
	\end{enumerate}
\end{theorem}
\begin{proof}
	Let $R = u_3u_2^{\ge 2} v_2 w_0 \overline{u_2}^{\ge 2} \overline{u_3}$, which is a regular language. 
	Under the assumption that none of the conditions 1-3 holds, $L := \HC_\alpha^*(w_0) \cap R = \{u_3u_2^i v_2 w_0 \overline{u_2}^i \overline{u_3} \mid i \ge 2\}$ holds.
	As mentioned previously, if the second condition does not hold, which is equivalent to $u_3 \neq u_2^2$, then $HC_\alpha^*(w_0 \overline{u_3})$ cannot contain any word in the above intersection. 
	Thus, $L = (\HC_\alpha^*(w_0 \overline{u_2}) \cap R) \cup (\HC_\alpha^*(v_2 w_0) \cap R)$. 
	Using Lemmas~\ref{lem:suf_relation_aprefs} and \ref{lem:factor_relation_shortests}, we can easily prove the emptiness of the second intersection of the above sum. 
	This check is left to the reader, and the authors recommend them to check at least $\HC_\alpha^*(v_2 w_0 \overline{u_3} \ \overline{v_2}) \cap R = \emptyset$ because this check involves the important fact that $\calpha \ \overline{u_2} \le_p \calpha \ \overline{u_3}$ implies $u_3 = u_2^2$ and causes a contradiction. 
	As a result, we have $L = \HC_\alpha^*(w_0 \overline{u_2}) \cap R$. 
	Informally speaking, in order to produce a word in $R$ from $w_0$, we first have to extend $w_0$ to the right by $\overline{u_2}$. 

	Now we can extend $w_0 \overline{u_2}$ to the right by $\overline{u_2}$ $i$-times to obtain $w_0 \overline{u_2}^i$. 
	If this obtained word is extended to the left, then the word will be in $u_2 \Sigma^* w_0 \Sigma^* \overline{u_2}$. 
	Let us check that $u_2 \Sigma^* w_0 \Sigma^* \overline{u_2} \cap u_3 \Sigma^* w_0 \Sigma^* \overline{u_3} = \emptyset$. 
	If the intersection is not empty, then $u_3 \alpha \le_p u_2 x \alpha$ for some $x \in \{u_2, u_3, v_2\}^+$. 
	Due to Lemma~\ref{lem:apref_casuf_apref}, $u_3 \in u_2 \{u_2, v_2\}^+$, but actually we can say $u_3 \in u_2 \{u_2, v_2\}$ for $u_3$ is the second shortest nonempty $\alpha$-prefix of $w_0$. 
	However, this means that either the condition 1 or 2 holds, and contradicts our assumption. 
	Thus, we have only one choice; extending $w_0 \overline{u_2}^i$ to the right by $\overline{u_3}$. 

	As mentioned above, $\calpha \ \overline{u_2} \le_p \calpha \ \overline{u_3}$ cannot hold so that we cannot extend $w_0 \overline{u_2}^i \overline{u_3}$ further to the right to obtain a word in $R$. 
	Thus, we should extend this word to the left either by $u_3 u_2^j$ for some $j \le i$ or by $u_3 u_2^i v_2$. 
	Lemmas~\ref{lem:suf_relation_aprefs} and \ref{lem:factor_relation_shortests} prove that the former choice will not lead us to any word in $R$. 
	Now it suffices to mention that extending $u_3 u_2^i v_2 w_0 \overline{u_2}^i \overline{u_3}$ further to the left because such an extension force the contradictory relation $\calpha \ \overline{u_2} \le_p \calpha \ \overline{u_3}$ to hold. 
\end{proof}

	\section{Conclusion}

In this paper, we focused on finding conditions that a word $w_0 \in \alpha \Sigma^* \cap \Sigma^* \calpha$ must satisfy so that its iterated hairpin completion $\HC_\alpha^*(w_0)$ is a regular language. 
We classified the set of all non-crossing words according to the number $m$ of occurrences of $\alpha$ and the number $n$ of occurrences of $\calpha$ on a given word. 
For the cases when $n = 1$ and when $m = n = 2$, we proved that the iterated hairpin completion of a non-crossing $(m, n)$-word is regular. 
We also found a necessary and sufficient condition under which the iterated hairpin completion of a non-crossing $(3, 2)$-word is regular. 
This approach can be generalized to arbitrary non-crossing $(m, n)$-words, with the cases $(m, 1)$ and $(2, 2)$ being the induction base of an inductive proof. 
Future works include considering the same problem for crossing-words. 
In this case, Lemma~\ref{lem:length_nonoverlap_apref_casuf} or Theorem~\ref{thm:nonx_initial_once} does not hold any more, and hence, it may get harder to analyze the derivation processes of how a word is obtained from a given word $w_0$ by iterated hairpin completion. 
In addition, we investigated only the case when the suffix of length $k$ of an initial word $w_0$ is the complement of its prefix of the same length, but we eventually have to consider $w_0$ in $\alpha \Sigma^* \cap \Sigma^* \overline{\beta}$, where $\beta$ might not be equal to $\alpha$ (double-primer hairpin completion). 
We can easily observe that one-step hairpin completion with respect to $\alpha$ ($\beta$) derives a word in $\beta \Sigma^* \cap \Sigma^* \overline{\beta}$ (resp.~$\alpha \Sigma^* \cap \Sigma^* \calpha$) from $w_0$. 
Thus, results obtained in this study of single-primer hairpin completion are important step towards this most general setting of the regularity test problem of iterated hairpin completion of a single word. 
Another direction of research is to consider stopper sequences as in Whiplash PCR \cite{HAKSY00, SKKGYISH99}.

	\bibliographystyle{plain}
	\bibliography{sekibib}

\end{document}